\documentclass[12pt,reqno,twoside]{article}

\usepackage{xcolor}
\usepackage{tikz-cd}
\usepackage{listings}
\usepackage[T1]{fontenc}
\usepackage{verbatim}
\usepackage{multicol}
\usepackage[]{amsmath, amssymb, amsthm, tabularx}

\usepackage[]{a4, xcolor, here}

\usepackage[]{pstricks, pst-text, pst-node, pst-tree, gastex}

\usepackage{tikz}

\usepackage[tikz]{bclogo}

\usepackage{mhsetup, mathtools}

\usepackage{boites,graphicx}

\usepackage{soul}

\definecolor{pastelyellow}{rgb}{0.99, 0.99, 0.59}
\definecolor{aqua}{rgb}{0.0, 1.0, 1.0} 
\definecolor{aquamarine}{rgb}{0.5, 1.0, 0.83} 
\definecolor{bananayellow}{rgb}{1.0, 0.88, 0.21}
\definecolor{burgundy}{rgb}{0.5, 0.0, 0.13}
\definecolor{ao(english)}{rgb}{0.0, 0.5, 0.0}

\setul{}{0.2ex}
\setulcolor{bananayellow}

\usepackage[normalem]{ulem}

\usepackage{mathrsfs}

\usepackage{stmaryrd}

\usepackage{enumerate}

\usepackage{paralist}

\usepackage{multienum}

\usepackage{multicol}

\usepackage{fancyhdr}

\usepackage{datetime}

\usepackage{everypage}
\usepackage[contents={},opacity=1,scale=1.6,
color=gray!90]{background}
\usepackage{ifthen}

\usepackage[]{hyperref} 

\hypersetup{
		colorlinks = true,
		linkcolor = red,
		anchorcolor = black,
		citecolor = blue, 
		filecolor = cyan,
		menucolor = red,
		runcolor = cyan,
		urlcolor = blue,
		linkbordercolor = {white},
		linktocpage = true
}

\newtheorem{theorem}{Theorem}[section]
\newtheorem{proposition}[theorem]{Proposition}

\theoremstyle{definition}

\newtheorem{remark}[theorem]{Remark}

\makeatletter
\def\thmhead@plain#1#2#3{%
  \thmname{#1}\thmnumber{\@ifnotempty{#1}{ }\@upn{#2}}%
  \thmnote{ {\the\thm@notefont#3}}}
\let\thmhead\thmhead@plain
\makeatother

\newcommand{\0}{\mathrm{0}}

\newcommand{\cC}{\mathcal{C}}

\newcommand{\cF}{\mathcal{F}}
\newcommand{\cG}{\mathcal{G}}
\newcommand{\cH}{\mathcal{H}}

\newcommand{\cP}{\mathcal{P}}

\newcommand{\cS}{\mathcal{S}}
\newcommand{\cU}{\mathcal{U}}
\newcommand{\cV}{\mathcal{V}}
\newcommand{\cW}{\mathcal{W}}

\newcommand{\rsp}[1]{{\mathrm{rowsp}{#1}}}

\newcommand{\bmid}{\hbox{$\,$\vrule height 8.5 pt depth 3.5 pt width 0.3pt$\,$}}

\newcommand{\bbN}{{\mathbb N}}

\newcommand{\bbF}{{\mathbb F}}

\renewcommand{\geq}{\geqslant}
\renewcommand{\leq}{\leqslant}

{\end{list}}%

\begin{document}

\renewcommand{\headrulewidth}{0pt}

\rhead{ }
\chead{\scriptsize Optimum Distance Flag Codes from Spreads via Perfect Matchings in Graphs}
\lhead{ }

\title{Optimum Distance Flag Codes from Spreads\\ via Perfect Matchings in Graphs
\renewcommand\thefootnote{\arabic{footnote}}\footnotemark[1] 
}

\author{\renewcommand\thefootnote{\arabic{footnote}}
Clementa Alonso-Gonz\'alez\footnotemark[1],\,  Miguel \'Angel Navarro-P\'erez\footnotemark[1], \\ 
\renewcommand\thefootnote{\arabic{footnote}} 
 Xaro Soler-Escriv\`a\footnotemark[2]}

\footnotetext[1]{Dpt.\ de Matem\`atiques, Universitat d'Alacant, Sant Vicent del Raspeig, Ap.\ Correus 99, E -- 03080 Alacant. \\ E-mail adresses: \texttt{clementa.alonso@ua.es, miguelangel.np@ua.es, xaro.soler@ua.es.}}

\date{\today}

\maketitle

\begin{abstract}
In this paper, we study flag codes on the vector space $\bbF_q^n$, being $q$ a prime power and $\bbF_q$ the finite field of $q$ elements. More precisely, we focus on flag codes that attain the maximum possible distance ({\em optimum distance flag codes})  and can be obtained from a spread of $\bbF_q^n$. We characterize the set of admissible type vectors for this family of flag codes and also provide a construction of them based on well-known results about perfect matchings in graphs. This construction attains both the maximum distance for its type vector and the largest possible cardinality for that distance.
\end{abstract}


\section{Introduction}\label{sec:Introduction}
Network coding is a part of coding theory awaking a lot of interest during the last years. Subspace codes were introduced for the first time in \cite{KoetKschi08} as adequate error-correction codes in random network coding, that is, when the network is non-coherent. These codes consist of families of subspaces of a given $n$-dimensional vector space over a finite field $\mathbb{F}_q$ endowed with a specific distance. If the dimension of all the subspaces is fixed, we say that the code is a {\em constant dimension code}. Due to their good properties, {\em spread codes} or just {\em spreads} have become one of the most important families of constant dimension codes (\cite{MangaGorlaRosen08, MangaTraut14}). Spreads are  objects widely studied in classical finite geometry without paying attention to their application to coding theory (see \cite{Segre64}).

In \cite{KoetKschi08}, the authors proposed a suitable single-source multicast network channel that is used only once, so that subspace codes can be considered as {\em one-shot codes}.  Using the channel more than once was suggested originally in \cite{NobUcho09} as an option when the field size $q$ and the packet size $n$ could not be increased. This gives rise to the so-called {\em multishot codes} or {\em multishot constant dimension codes} if the dimension at each shot is constant. In a multishot constant dimension code $\cC$, codewords consist of sequences of $r$ subspaces of $\mathbb{F}_q^n$ with respective fixed dimensions $t_1,\dots,t_r$, sent in $r$ successive shots. In particular, if $0<t_1<\dots<t_r<n$ and these subspaces are nested, we have \emph{flags} on $\mathbb{F}_q^n$. In that case, we say that $\cC$ is a \emph{flag code of type} $(t_1,\dots,t_r)$ on $\mathbb{F}_q^n$. A channel model for flag codes in the setting of network coding was also described in \cite{LiebNebeVaz18}.

In this paper we focus on the construction of flag codes of type $(t_1,\dots,t_r)$ on $\mathbb{F}_q^n$ with the maximum possible distance ({\em optimum distance flag codes}) such that the subspace code sent at some shot is a $k$-spread, being $k$ a divisor of $n$. We show that this construction is not possible in general and we conclude that an {\em admissible type vector} $(t_1,\dots,t_r)$ (for $k$) has to satisfy the following: $k \in \{t_1,\dots,t_r\} \subseteq \{1,\dots,k,n-k,\dots,n-1\}$, that is, the dimension $k$ must appear in the type vector whereas no dimension between $k+1$ and $n-k-1$ is allowed. In \cite{CasoPlanar}, it was proved that optimum distance full flag codes (type $(1,\dots,n-1)$) such that the subspace code sent at some shot is a $k$-spread can be constructed just if $n=2k$ or $n=3$ and $k=1$. In these two situations any full type vector is admissible. For $n=2k$, in \cite{CasoPlanar} it is also described a concrete construction of optimum distance full flag codes from planar spreads using the fact that $k=n-k$. However, this equality does not hold for a general admissible type vector and we need to overcome the gap between dimensions $k$ and $n-k$. We solve this problem by using the existence of perfect matchings in a specific bipartite regular graph. We give our construction gradually, starting with the type $(1, n-1)$ (this includes the full type vector for the case $n=3, k=1$), following with the type $(k, n-k)$, to finish with the full admissible type $(1,\dots, k,n-k,\dots, n-1)$. Our construction provides codes with the maximum distance for the given admissible type vector and attains the largest cardinality among optimum distance flag codes of its type.

The paper is organized as follows: In Section \ref{sec:Preliminaries} we recall some background on finite fields, constant dimension codes, flag codes and graphs. In Section \ref{sec: optimum distance flag codes from spreads} we characterize first the admissible type vectors to have a flag code with the maximum possible distance and a spread as the subspace code used at some shot. Then, we undertake the construction of our codes in several stages: we consider first the type $(1, n-1)$ and construct optimum distance flag codes from the spread of lines, exploiting the existence of perfect matchings in bipartite regular graphs. Then, using the field reduction map, we translate the previous construction to the type $(k, n-k)$. Finally, by taking advantage of some properties satisfied by the mentioned map, we finish with the full admissible type $(1,\dots, k,n-k,\dots, n-1)$ and any other admissible type vector. We show that our codes have the best size for the given admissible type vector and the associated maximum distance.  We complete this section with an example of our construction to obtain an optimum distance flag code of type $(2,4)$ on $\bbF_2^6$ having a $2$-spread as the subspace code used at the first shot.

\section{Preliminaries}\label{sec:Preliminaries}

We devote this section to recall some background  we will need along this paper. This background involves finite fields, subspace and flag codes and graph theory.  

\subsection{Results on finite fields}\label{subsec:results on finite fields}

Most of the following definitions and results about finite fields as well as the corresponding proofs can be found in \cite{LidNiede94}. 

Let $q$ be a prime power and $\bbF_q$ the finite field with $q$ elements. Consider $f(x) \in \bbF_q[x]$ a monic irreducible polynomial of degree $k$ and $\alpha \in \bbF_{q^k}$ a root of $f(x)$. Then we have that $\bbF_{q^k} \cong \bbF_q[\alpha]$, which allows us to realize the field $\bbF_{q^k}$ as $\bbF_q[\alpha]$. If $f(x)=x^k+\sum_{i=0}^{k-1}a_ix^i \in \bbF_q[x]$, the following square matrix
$$P=
\begin{pmatrix}
 0       & 1       &  0       &  \dots  & 0        \\
  0      & 0       &  1       &  \dots  & 0        \\
  \vdots & \vdots  &  \vdots  &  \ddots & \vdots   \\
   0     & 0       &  0       &  \dots  & 1        \\
  -a_0   & -a_1    &  -a_2    &  \dots  & -a_{k-1}
\end{pmatrix}
$$
is called the \emph{companion matrix} of $f(x)$ and satisfies that 
$\bbF_q[\alpha] \cong \bbF_q[P]$. Then, $ \bbF_q[P]$ is a field with $q^k$ elements. We also have the natural field isomorphism
\begin{equation}\label{eq:field isomorphism}
 \phi:\bbF_{q^k}\rightarrow \bbF_q[P], \qquad \sum_{i=0}^{k-1}v_i\alpha^i \mapsto \sum_{i=0}^{k-1}v_i P^i. 
\end{equation}

For any positive integer $n$, we denote by $\cP_q(n)$ the \emph{projective geometry} of $\bbF_q^n$, that is, the set of all vector subspaces of $\bbF_q^n$. The \emph{Grassmann variety} $\cG_q(k,n)$ is the set of all $k$-dimensional subspaces of $\bbF_q^n$. Any subspace $\cU \in \cG_q(k,n)$ can be generated by the rows of some full-rank matrix $U \in \bbF_q^{k\times n}$. In that case, we write $\cU=\rsp(U)$ and say that $U$ is a \emph{generator matrix} of $\cU$. By taking the generator matrix in reduced row echelon form (RREF) we get uniqueness in the matrix representation of the subspace $\cU$.

Let us take $n=ks$ with $k>1$. The field isomorphism $\phi$ provided by (\ref{eq:field isomorphism}), in turn, naturally induces a map $\varphi$ between  $\cP_{q^k}(s)$ and $\cP_q(ks)$ given by:

\begin{equation}\label{eq: field reduction on m-subspaces} 
\begin{array}{rccc}
{\varphi}: & \qquad \cP_{q^k}(s)                             & \longrightarrow &  \cP_q(ks)\\
 &      &   &  \\
              & \mathrm{rowsp}\begin{pmatrix}
                  x_{11}&\dots&x_{1s}\\
                  \vdots &\ddots& \vdots\\
                  x_{m1}&\dots&x_{ms}
                  \end{pmatrix}               & \longmapsto     & \mathrm{rowsp}\begin{pmatrix}
                  \phi(x_{11})&\dots&\phi(x_{1s})\\
                  \vdots &\ddots& \vdots\\
                  \phi(x_{m1})&\dots&\phi(x_{ms})
                                                                        \end{pmatrix}.
              \end{array}
\end{equation}
This map is known as \emph{field reduction} since it maps subspaces over $\bbF_{q^k}$ into subspaces over the subfield $\bbF_q$ (see \cite{GoManRo12, LavVoorde15, MangaGorlaRosen08, MangaTraut14}). Let us recall some useful properties of the map $\varphi$ pointed out in \cite{LavVoorde15} that we will use in Section \ref{subsect:(k,n-k) type}.

\begin{proposition}\label{prop: field reduction}
The map $\varphi$ defined in (\ref{eq: field reduction on m-subspaces}) satisfies the following:
\begin{enumerate}
\item $\varphi$ is injective.
\item For any pair of subspaces  $\cU, \cV $ of $\mathbb{F}_{q^k}^s$, we have $\varphi(\cU \cap \cV)=\varphi(\cU) \cap \varphi(\cV)$.
\item Given $\cU, \cV$ subspaces of $ \bbF_{q^k}^s$ with $\cU\subseteq \cV$, then $\varphi(\cU) \subseteq \varphi(\cV)$.
\item For any $m\in\{1,\dots,s-1\}$, it holds that $\varphi(\cG_{q^k}(m,s))\subseteq \cG_{q}(mk,sk)$.

\end{enumerate}
\end{proposition}

\subsection{Constant dimension codes}\label{subsec:constant dimension codes}

The Grassmannian $\cG_q(k,n)$ can be considered as a metric space with the \emph{subspace distance} defined as 
\begin{equation}\label{def: subspace distance}
d_S(\cU, \cV)= \dim(\cU+\cV)-\dim(\cU\cap\cV)=2(k-\dim(\cU\cap\cV)),
\end{equation}
for all $\cU, \cV \in \cG_q(k,n)$ (see \cite{KoetKschi08}).

A \emph{constant dimension (subspace) code} of dimension $k$ and length $n$ is any nonempty subset $\cC \subseteq \cG_q(k,n)$. The \emph{minimum subspace distance} of the code $\cC$ is defined as 
$$
d_S(\cC)=\min\{ d_S(\cU, \cV) \ | \ \cU, \cV \in \cC, \ \cU \neq \cV \}
$$
(see \cite{TrautRosen18} and references therein, for instance). It follows that the minimum distance of a constant dimension code $\cC$ is upper-bounded by
\begin{equation}\label{eq: bound subspace distance}
d_S(\cC)\leq
\left\lbrace
\begin{array}{lll}
2k      & \text{if} & 2k\leq n, \\
2(n-k)  & \text{if} & 2k > n.
\end{array}
\right.
\end{equation}
 Constant dimension codes $\cC\subseteq \cG_q(k,n)$ in which the distance between any pair of different codewords is $d_S(\cC)$ are said to be \emph{equidistant}. For such codes, there exists some value $c<k$ such that, given two different subspaces $\cU, \cV \in \cC$, it holds that $\dim(\cU\cap\cV)=c$. Hence, the minimum distance of the code is precisely $d_S(\cC)=2(k-c)$, and $\cC$ is also called an \emph{equidistant $c$-intersecting} constant dimension code. In case the value $c$ is the minimum possible dimension of the intersection between $k$-dimensional subspaces of $\bbF_q^n$, that is, 
$$
c=\left\lbrace
\begin{array}{lll}
0     & \text{if} & 2k\leq n, \\
2k-n  & \text{if} & 2k > n,
\end{array}
\right.
$$ equidistant $c$-intersecting codes attain the bound given in (\ref{eq: bound subspace distance}). In particular, for dimensions $k \leq \lfloor\frac{n}{2}\rfloor$, we have that these codes are $0$-intersecting codes known as \emph{partial spreads}. The cardinality of any partial spread $\cC$ in $\cG_q(k,n)$ is upper bounded by
\begin{equation}\label{eq: bound partial spreads}
|\cC| \leq \left\lfloor \frac{q^n-1}{q^k-1}  \right\rfloor.
\end{equation}
Whenever $k$ divides $n$, this bound is attained by the so-called \emph{spread codes} (or $k$-\emph{spreads}) of $\bbF_q^n$.
Notice that a $k$-spread $\cS$ is a subset of $\cG_q(k,n)$ whose elements give a partition of $\bbF_q^n$. Spreads are classical objects coming from Finite Geometry (see \cite{Segre64}, for instance). For further information related to spreads in the network coding framework, we refer the reader to \cite{GoManRo12, MangaGorlaRosen08, MangaTraut14,TrautRosen18}.

The following spread is due to Segre \cite{Segre64}. In the network coding setting, it was presented for the first time in \cite{MangaGorlaRosen08} as a construction of spread code. Denote by $GL_k(q)$ the general linear group of degree $k$ over the field $\bbF_q$. Let $P \in  GL_k(q)$ be the companion matrix of a monic irreducible polynomial in $\bbF_q[x]$. We will write $I_k$ and $\0_k$ to denote the identity matrix and the zero matrix of size $k\times k,$ respectively. Take $s \in \bbN$ such that $n=sk$. Then, the following family of $k$-dimensional subspaces is a spread code:
\begin{equation}\label{eq:ManganielloSpreadDefinition}
\cS(s,k,P)=\{\rsp(S)\mid S\in \Sigma \} \subseteq \cG_q(k,n),
\end{equation}
where $\Sigma$ is the set of $k\times ks$ matrices

\begin{equation}\label{eq:matrices for the spread}
 \Sigma=\{( A_1|A_2|\dots|A_s)\,\big|\,  A_i \in \bbF_q[P]\}    
\end{equation}
with the first non-zero block from the left equal to $I_k$.

\begin{remark}\label{rem:bijection lines and spread}
Notice that the matrices in $\Sigma$ are in reduced row echelon form and it is clear that the field reduction map $\varphi$ defined in (\ref{eq: field reduction on m-subspaces}) gives a bijection between the Grassmannian of lines $\cG_{q^k}(1,s)$ and the spread code $\cS(s,k,P)$
\begin{equation}\label{eq: bijection spread lines and spread} 
\begin{array}{rccc}
{\varphi}{\bmid_{\cG_{q^k}(1,s)}}: & \qquad \cG_{q^k}(1,s)                            & \longrightarrow & \cS(s,k,P)  \\
 &      &   &  \\
              & \mathrm{rowsp}\begin{pmatrix}
                  x_{11},\dots,x_{1s}
        
                  \end{pmatrix}               & \longmapsto     & \mathrm{rowsp}\begin{pmatrix}
                  \phi(x_{11})|\dots|\phi(x_{1s})
                                                                        \end{pmatrix}.
              \end{array}
\end{equation} We will come back to this fact in Section \ref{subsect:(k,n-k) type}.

\end{remark}
Given a constant dimension code $\cC \subseteq \cG_q(k,n)$, the \emph{dual code} of $\cC$ is the subset of $ \cG_q(n-k,n)$ given by 
$$
\cC^\perp =\{ \cU^\perp \ | \ \cU \in \cC\},
$$where $\cU^\perp$ is the orthogonal of $\cU$ with respect to the usual inner product in $\bbF_q^n$. In \cite{KoetKschi08}, it was proved that $\cC$ and $\cC^\perp$ have both the same cardinality and minimum distance. Notice that the dual of a partial spread of dimension $k\leq \lfloor \frac{n}{2}\rfloor$ is an equidistant $(n-2k)$-intersecting code of dimension $n-k$ and conversely.

\subsection{Flag codes}
Subspace codes were introduced for the first time in \cite{KoetKschi08} as error-correction codes in random network coding. In that paper, the authors propose a suitable network channel with a single transmitter and several receivers that is used just once, so that subspace codes can be considered as {\em one-shot} codes. The use of the channel more than once was suggested originally in \cite{NobUcho09} and gave rise to the so-called multishot codes as a generalization of subspace codes. We call \emph{multishot code} of length $r\geq 1$, or \emph{$r$-shot code}, to any nonempty subset $\cC$ of $\cP_q(n)^r$.  In particular, if codewords in $\cC$  are sequences of nested subspaces, we say that $\cC$ is a \emph{flag code}. Flag codes were first studied as orbits of group actions in \cite{LiebNebeVaz18} and, in \cite{Kurz20}, the reader can find a study of bounds on the cardinality of full flag codes with a prescribed distance. Let us recall some concepts in the setting of flag codes.

A {\em flag} of type $(t_1, \dots, t_r)$, with $0<t_1<\dots <t_r<n$, on the vector space $\mathbb{F}_q^n$ is a sequence of subspaces $\mathcal{F}=(\mathcal{F}_1,\dots,  \mathcal{F}_r)$ in $\mathcal{G}_q(t_1,n) \times \dots \times \mathcal{G}_q(t_r,n) \subseteq \mathcal{P}_q(n)^r$ such that
$$
\{0\}\subsetneq \mathcal{F}_1 \subsetneq \dots \subsetneq \mathcal{F}_r \subsetneq \mathbb{F}_q^n.
$$
With this notation, $\mathcal{F}_i$ is said to be the {\em $i$-th subspace} of $\cF$. In case the type vector is $(1, 2, \dots, n-1),$ we say that ${\cF}$ is a {\em full flag}.

The space of flags of type $(t_1, \dots, t_r)$ on $\mathbb{F}_q^n$ is denoted by $\mathcal{F}_q((t_1,...,t_r),n)$ and can be endowed with the \emph{flag distance} $d_f$ that naturally extends the subspace distance defined in (\ref{def: subspace distance}):  given two flags $\cF=(\mathcal{F}_1,\dots,  \mathcal{F}_r)$ and $\cF'=(\mathcal{F}'_1,\dots,  \mathcal{F}'_r)$ in $\mathcal{F}_q( (t_1, \dots, t_r),n)$, the flag distance between them is
$$
d_f(\cF,\cF')= \sum_{i=1}^r d_S(\mathcal{F}_i, \mathcal{F}'_i).
$$

A \emph{flag code} of type $(t_1, \dots, t_r)$ on $\bbF_q^n$ is defined as any non-empty subset $\cC\subseteq \cF_q((t_1, \dots, t_r), n)$. The {\em minimum distance} of a flag code $\cC$ of type $(t_1, \dots, t_r)$ on $\bbF_q^n$ is given by
$$
d_f(\cC)=\min\{d_f(\cF,\cF')\ |\ \cF,\cF'\in\cC, \ \cF\neq \cF'\}.
$$
Given a type vector $(t_1, \dots, t_r)$, for every $i=1, \dots, r$, we define the \emph{$i$-projection} to be the map 
\begin{equation}\label{eq: i-projection} 
\begin{array}{rccc}
{p_i}: & \cF_q((t_1, \dots, t_r), n)  & \longrightarrow & \mathcal{G}_q(t_i,n)  \\
& & &\\
 &  \cF=(\cF_1,\dots ,\cF_r)   & \longmapsto   & p_i(\cF)= \cF_i.
\end{array}
\end{equation}

The \emph{$i$-projected code} of $\cC$ is the set $\mathcal{C}_i=\{p_i(\cF)\,| \,\cF \in \cC\}$. By definition, this code $\cC_i$ is a constant dimension code in the Grassmannian $\cG_q(t_i, n)$ and its cardinality satisfies $\vert \cC_i\vert \leq \vert \cC \vert$. We say that  $\cC$ is a \emph{disjoint} flag code if  $|\cC_1|=\dots=|\cC_r|=|\cC|$, that is, the $i$-projection $p_i$ is injective for any $i \in\{1, \dots, r\}$.

The distance of a flag code $\cC$ of type $(t_1, \dots, t_r)$ is upper bounded by
\begin{equation}\label{eq:quotamaxdistflag}
 d_f(\cC) \leq 2 \left( \sum_{t_i\leq \lfloor \frac{n}{2} \rfloor} t_i + \sum_{t_i > \lfloor \frac{n}{2} \rfloor} (n-t_i) \right).
\end{equation}
In particular, if $\cC$ is a full flag code, we have that  (\ref{eq:quotamaxdistflag}) becomes 
\[
d_f(\mathcal{C})\leq 
\left\lbrace
\begin{array}{lccc}
\dfrac{n^2}{2}, & \text{for} & n & \text{even}, \\ 
                &            &   &         \\
\dfrac{n^2-1}{2}, & \text{for} & n & \text{odd}.
\end{array}
\right.
\]

\subsection{Matchings in graphs}\label{subsec: Matchings in bipartite graphs}
Now we introduce some basic concepts and results on graphs in order to use them in the construction of a specific family of flag codes with the maximum distance in Section \ref{sec: optimum distance flag codes from spreads}. All these definitions and results together with their proofs can be found in \cite{Diestel2005}.

A \emph{graph} $G=(V,E)$ consists of a \emph{vertex set} $V$ and an \emph{edge set} $E\subset V \times V$ where an edge is an unordered pair of vertices. Two vertices $v,v'\in V$  are \emph{adjacent} if $(v,v') \in E$. Also,  we say that $(v,v')$ is an \emph{incident} edge with $v$ and $v'$. Two edges are \textit{adjacent} if they have a common vertex. Given a vertex $v \in V$ we call the \emph{degree} of $v$ to the number of incident edges with $v$. A graph $G$ is said to be \emph{$k$-regular}, if each vertex in $G$ has degree $k$.

On the other hand, a set of vertices (or edges) is \textit{independent} if it does not contain adjacent elements. A set $M\subseteq E$ of independent edges of a graph $G=(V,E)$ is called a \emph{matching}. A matching $M$ matches $S \subseteq V$ if every vertex in $S$ is incident with an edge in $M$ and $M$ is \emph{perfect} if it matches $V$.

A graph $G$ is \emph{bipartite} if the vertex set can be partitioned into two sets $V=A\cup B$ such that there is no pair of adjacent vertices neither in $A$ nor in $B$. For this class of graphs, perfect matchings are just bijections between $A$ and $B$ given by a subset of edges of the graph connecting each vertex in $A$ with another vertex in $B$.
The following classic result whose proof can be found in \cite{Diestel2005} (pages $37-38$) states the existence of perfect matchings in a family of graphs:

\begin{theorem}\label{theo:perfect matching}
Any $k$-regular bipartite graph admits a perfect matching.
\end{theorem}

This theorem will be used through Section \ref{sec: optimum distance flag codes from spreads} to give perfect matchings of a particular regular bipartite graph of our interest. Such matchings will allow us to construct disjoint flag codes of a specific type as we will show later.

\section{Optimum distance flag codes from spreads}\label{sec: optimum distance flag codes from spreads}

Flag codes attaining the bound in (\ref{eq:quotamaxdistflag}) are called \emph{optimum distance flag codes} and can be characterized in terms of their projected codes in the following way: 

\begin{theorem}(see \cite{CasoPlanar})\label{theo:characterization optimum distance}
Let $\cC$ be a flag code of type $(t_1, \dots, t_r)$. The following statements are equivalent:
\begin{enumerate}
\item[(i)] $\cC$ is an optimum distance flag code.
\item[(ii)] $\cC$ is disjoint and every projected code $\cC_i$ attains the maximum possible subs\-pace distance.
\end{enumerate}
\end{theorem}

As a consequence, the $i$-projected codes of an optimum distance flag code have to be partial spreads if $t_i \leq \lfloor \frac{n}{2} \rfloor$ and equidistant $(2t_i-n)$-intersecting subspace codes for dimensions $t_i> \lfloor \frac{n}{2} \rfloor$.

As mentioned in Section \ref{subsec:constant dimension codes}, whenever $k$ divides $n$, $k$-spread codes are partial spread codes (constant dimension codes with maximum distance) with the best size. This good property of spreads naturally gives rise to the question of fin\-ding  optimum distance flag codes having a spread as their $i$-projected code when the dimension $t_i$ is a divisor of $n$. Note that, due to the disjointness property, we could have at most one spread among the projected codes. In \cite{CasoPlanar} it was proved that having a spread as a projected code makes optimum distance flag codes attain the maximum possible size as well.

\begin{theorem}\cite[Theorem 3.12]{CasoPlanar}\label{theo: maximum cardinality iff spread}
Let $k$ be a divisor of $n$ and assume that $\cC$ is an optimum distance flag code of type $(t_1, \dots, t_r)$ on $\bbF_q^n$. If some $t_i=k$, then $|\cC|\leq \frac{q^n-1}{q^k-1}$ and equality holds if, and only if, the $i$-projected code $\cC_i$ is a $k$-spread of $\bbF_q^n$. 
\end{theorem}

In the same paper, it is shown that, for the full type vector $(1,\dots,n-1)$, it is possible to find optimum distance full flag codes having a spread as a $k$-projected code only if either $n=2k$ or $n=3$ and $k=1$. Observe that there, the authors always work with full flag codes (the full type vector is fixed) and then provide conditions on $n$ and $k$. Now, we deal with the inverse problem: given $n$ and a divisor $k$ of $n$, we look for conditions on the type vector of an optimum distance flag code on $\bbF_q^n$ having a $k$-spread as a projected code. We conclude that not all the type vectors are allowed. Let us describe the admissible ones and provide a construction of optimum distance flag codes for them, based on the existence of perfect matchings in a specific graph.

\subsection{Admissible type vectors}

This paper is devoted to explore the existence of optimum distance flag codes of a general type vector $(t_1,\dots,t_r)$, not necessarily the full type, having a spread as their $i$-projected code when $t_i$ is a divisor of $n$. Next result states the necessary conditions that the type vector $(t_1,\dots,t_r)$ must satisfy.

\begin{theorem}
\label{prop: type (1,...,k, n-k, ..., n-1)}
Let $\cC$ be an optimum distance flag code of type $(t_1, \dots, t_r)$on $\bbF_q^n$. Assume that some dimension $t_i=k$ divides $n$ and  the associated projected code $\cC_i$ is a $k$-spread. Then, for each $j\in\{1, \dots, r\}$, either $t_j\leq k$ or $t_j \geq n-k$.

\end{theorem}
\begin{proof}

Notice that in case $i=r$, clearly $t_j\leq t_r=k$, for every $j=1, \dots, r$. Suppose that $i<r$. Let us show that $t_{i+1}\geq n-k$.

Since $t_i=k$ divides $n$, we can write $n=sk$, for some $s\geq 2$. If $s=2$, we have that $n-k=k$ and the result trivially holds. In case $s>2$, then $s<2(s-1)$ and we have that $2k < n < 2(s-1)k=2(n-k)$. We deduce that $k \leq \lfloor \frac{n}{2} \rfloor < n-k$.
Now, by contradiction, assume that $t_{i+1}< n-k$. We distinguish two possibilities:
\begin{enumerate}
\item If $k < t_{i+1} \leq  \lfloor \frac{n}{2}\rfloor$, since $\cC$ is an optimum distance flag code, by  Theorem \ref{theo:characterization optimum distance}, its projected code $\cC_{i+1}$ must be a partial spread of dimension $t_{i+1}$ and cardinality $\vert \cC_{i+1}\vert=\vert \cC_{i}\vert=\frac{q^{n}-1}{q^{k}-1}$. Contradiction with (\ref{eq: bound partial spreads}).
\item If $\lfloor \frac{n}{2}\rfloor < t_{i+1} <n-k$, the projected code $\cC_{i+1}$ has to be an equidistant $(2t_{i+1}-n)$-intersecting constant dimension code. In other words, the subspace distance of $\cC_{i+1}$ is $2(n-t_{i+1})$. Hence, its dual code $\cC_{i+1}^{\perp}$ is a partial spread of dimension $n-t_{i+1}>n-k>k$ and cardinality  $\vert \cC_{i+1}\vert=\vert \cC_{i}\vert=\frac{q^{n}-1}{q^{k}-1}$, which again contradicts (\ref{eq: bound partial spreads}).
\end{enumerate}
We conclude that $t_{i+1} \geq n-k$.
\end{proof}

\begin{remark}
This result provides a necessary condition on the type vector of any optimum distance flag code on $\bbF_q^n$ having a $k$-spread as a projected code: clearly the dimension $k$ must appear in the type vector but no dimension between $k+1$ and $n-k+1$ can be part of it. Notice that every type vector containing the dimension $k$ is admissible when $n=2k$ since, in that case, it holds $k=n-k$. This particular case has been already studied in \cite{CasoPlanar}, where it was proved that optimum distance flag codes of any type vector containing the dimension $k$ can be constructed from a $k$-spread (planar spread). Moreover, those codes were shown to attain the maximum possible cardinality as well. In the next subsection, we tackle the problem of constructing flag codes attaining the maximum distance and having a $k$-spread as their projected code for any admissible type vector in the general case  $n=ks$, for $s\geq 3$.
\end{remark}

\subsection{A construction based on perfect matchings}\label{sec:our construction}
This part is devoted to describe a specific construction of optimum distance flag codes on $\bbF_q^n$ from a $k$-spread of a given admissible type vector $(t_1, \dots, t_r)$. By means of Theorem \ref{prop: type (1,...,k, n-k, ..., n-1)}, if such codes exist, their type vector must satisfy $k \in \{t_1, \dots, t_r\} \subseteq \{1, \dots, k,  n-k, \dots, n-1\}$. For the sake of simplicity, we undertake this construction in several phases: we consider first the admissible type vector $(1,n-1)$, that is, the construction of optimum distance flag codes from the spread of lines. Secondly, by using the field reduction map defined in Section \ref{subsec:results on finite fields}, we properly translate the construction in the first step to get optimum distance flag codes of type vector $(k,n-k)$ having the $k$-spread $\cS$ introduced in (\ref{eq:ManganielloSpreadDefinition}) as its first projected code. Then, taking advantage of certain properties of the $k$-spread $\cS$, we extend the construction in the second step to obtain optimum distance flag codes of the \emph{full admissible type}, that is, $(1, \dots, k,  n-k, \dots, n-1)$. Finally, this last construction gives optimum distance flag codes of any admissible type vector after a suitable \emph{puncturing process}. Let us explain in detail all these stages.

\subsubsection{The type vector $(1,n-1)$: starting from the spread of lines}\label{subsec: case k=1}

Take $n\geq 3$. In this section we provide a construction of optimum distance flag codes on $\bbF_q^n$ from the spread of lines, that is, having the Grassmannian $\cG_q(1, n)$ as a projected code. By Theorem \ref{prop: type (1,...,k, n-k, ..., n-1)}, the only admissible type vector in this case is $(1,n-1)$. In other words, to give an optimum distance flag code from the spread of lines of $\bbF_q^n,$ we have to provide a family of $|\cG_q(1, n)|$ pairwise disjoint flags of length two, all of them consisting of a line contained in a hyperplane. To do so, we translate this problem to the one of finding perfect matchings in bipartite regular graphs, using the results given in Section \ref{subsec: Matchings in bipartite graphs}. Let us precise this.

Consider the graph $G=(V, E)$, with set of vertices $V=\cG_q(1,n) \cup \cG_q(n-1,n)$ and set of edges $E$ defined by
$$
E=\{(l,H)\in \cG_q(1,n) \times \cG_q(n-1,n)  \ | \ l\subset H \}.
$$
Notice that the set of vertices in $G$ consists of the lines and hyperplanes of $\bbF_q^n$. An edge $(l, H)$ of $G$  exists if, and only if, the line $l$ is contained in the hyperplane $H$. With this notation, next result holds.

\begin{proposition}
The graph $G=(V, E)$ is bipartite and $\frac{q^{(n-1)}-1}{q-1}$-regular.
\end{proposition}
\begin{proof}
It is clear that $G$ is a bipartite graph by definition. Moreover, the number of hyperplanes containing a fixed line coincides with the number of lines lying on a given hyperplane. This number is precisely $\frac{q^{(n-1)}-1}{q-1}$. Then, the degree of any vertex in $G$  coincides with this value and then $G$ is $\frac{q^{(n-1)}-1}{q-1}$-regular.
\end{proof}

Note that the problem of giving a family of flags with the desired conditions can be seen as the combinatorial problem of giving a perfect matching in $G$. Since $G$ is a regular bipartite graph, we can use Theorem \ref{theo:perfect matching} to conclude that there exist perfect matchings in $G$. More precisely, there exists a subset $M\subset E$ that matches $V$, that is, each edge in $M$ has an extremity in $\cG_q(1,n)$ and the other one in $\cG_q(n-1,n)$. In particular, the set $M$ has a number of edges equal to $|\cG_q(1,n)|$. This matching $M$ induces naturally a bijection, also denoted by $M$, between the set of lines and the set of hyperplanes in $\mathbb{F}_q^n$. Moreover, by the definition of $E$, we have that the map  $M : \cG_q(1,n) \rightarrow \cG_q(n-1,n)$ satisfies that $l\subset M(l)$ for any $l\in \cG_q(1,n)$. This fact allows us to construct a family of flags of type $(1,n-1)$ on $\bbF_q^n$ in the following way:
\begin{equation}\label{def: optimum distance flag code on spread lines}
\widetilde{\cC}=\widetilde{\cC}_M=\{ (l,M(l)) \ | \ l \in \cG_q(1,n) \}.
\end{equation}
Let us see that the family $\widetilde{\cC}$ is a flag code  with projected codes $\widetilde{\cC}_1=\cG_q(1,n)$ and $\widetilde{\cC}_2=\cG_q(n-1,n)$ satisfying the desired conditions.
\begin{theorem}\label{theo:optimum distance (1,n-1)}
Given $n\geq 3$, the code $\widetilde{\cC}$ defined in (\ref{def: optimum distance flag code on spread lines}) is an optimum distance flag code of type $(1,n-1)$ on $\bbF_q^n$ with the spread of lines as a projected code.
\end{theorem}
\begin{proof}
Since the map $M$ defined above is bijective, the code $\widetilde{\cC}$ must be a disjoint flag code with projected codes $\widetilde{\cC}_1=\cG_q(1,n)$ and $\widetilde{\cC}_2=\cG_q(n-1,n).$ In particular, as $d_S(\widetilde{\cC}_1)=d_S(\widetilde{\cC}_2)=2$ is the maximum possible distance for constant dimension codes of dimension $1$ and $n-1$ in $\bbF_q^n$. By Theorem \ref{theo:characterization optimum distance}, we have that $\widetilde{\cC}$ is an optimum distance flag code with $\cG_q(1,n)$ as a projected code.
\end{proof}
\begin{remark}
Observe that, by means of Theorem \ref{theo: maximum cardinality iff spread}, our code $\widetilde{\cC}$ defined as above, attains the maximum possible cardinality for flag codes of type $(1,n-1)$ and distance $4$, which is 
$$
|\widetilde{\cC}| = \dfrac{q^n-1}{q-1}= q^{n-1}+q^{n-2}+\dots + q+1.
$$
For the particular case $n=3$, the previous bound was given in \cite{Kurz20}, where the author studied bounds for the cardinality of full flag codes with a given distance. Observe that this is the only case in which optimum distance full flag codes with a spread as a projected code can be constructed, apart from the case $n=2k$, already studied in \cite{CasoPlanar}.
\end{remark}

Note that, despite the fact that Theorem \ref{theo:perfect matching} guarantees the existence of perfect matchings in regular bipartite graphs, in order to provide an example of optimum distance flag codes of type $(1,n-1)$ on $\bbF_q^n$, we need to exhibit a precise matching in $G$. The problem of finding perfect matchings is classical in Graph Theory. The reader can find an algorithm to construct perfect matchings in the general setting of regular bipartite graphs in \cite{Diestel2005} and \cite{GoKapKhan13}. In Section \ref{subsec:example} we give an example of a flag code constructed from a perfect matching generated by an algorithm programmed in GAP and adapting the procedure described in \cite{Diestel2005} to our specific problem.

\subsubsection{The type vector $(k,n-k)$}\label{subsect:(k,n-k) type}
Take $n=ks$ a natural number with $k\geq 2$ and $s\geq 3$. In order to construct optimum distance flag codes of type $(k, n-k)$ on $\bbF_q^{n}$, we will use the construction of optimum distance flag codes of type $(1,s-1)$ on $\bbF_{q^k}^s$ given in Section \ref{subsec: case k=1} together with the field reduction map defined in Section \ref{subsec:results on finite fields}. Let us explain this construction.

Let $M : \cG_{q^k}(1,s) \rightarrow \cG_{q^k}(s-1,s)$ be a bijection such that $l\subset M(l)$ for any $l\in \cG_{q^k}(1,s)$. By Theorem \ref{theo:optimum distance (1,n-1)} we know that the code
$$\widetilde{\cC}=\widetilde{\cC}_M=\{ (l,M(l)) \ | \ l \in \cG_{q^k}(1,s) \}$$
is an optimum distance flag of type $(1, s-1)$ on $\bbF_{q^k}^s$. In particular, the code $\widetilde{\cC}$ is disjoint.  On the other hand, given $P \in  GL_k(q)$ the companion matrix of a monic irreducible polynomial of degree $k$ in $\bbF_q[x]$, the associated field isomorphism $ \phi:\bbF_{q^k}\rightarrow \bbF_q[P]$ induces the field reduction ${\varphi}: \cP_{q^k}(s)                              \longrightarrow  \cP_q(ks)$ as in (\ref{eq: field reduction on m-subspaces}).
Notice that, by Proposition \ref{prop: field reduction}, we have that for any $m\in\{1,\dots,s-1\}$, it holds that $\varphi(\cG_{q^k}(m,s))\subseteq \cG_{q}(mk,sk)$. Moreover, given $\cU, \cV$ subspaces of $ \bbF_{q^k}^s$ with $\cU\subseteq \cV$, then $\varphi(\cU) \subseteq \varphi(\cV)$. As a consequence, if $(l, M(l)) \in \widetilde{\cC}$, then $(\varphi(l),\varphi(M(l)))$ is a flag of type $(k, n-k)$ on $\bbF_{q}^n$.  This fact allows us to define a family of flags
over $\bbF_{q}^n$ as follows
\begin{equation}\label{eq:code(k,n-k)}
\widehat{\cC}=\{ (\varphi(l),\varphi(M(l))) \ | \ l \in \cG_{q^k}(1,s) \}.
\end{equation}

By Remark \ref{rem:bijection lines and spread} we know that $\varphi$ gives a bijection between $\cG_{q^k}(1,s)$ and the $k$-spread  
$\cS=\cS(s,k,P)$ defined in (\ref{eq:ManganielloSpreadDefinition}). Hence,  the family $\widehat{\cC}$ is a flag code with projected codes $$\widehat{\cC}_1=\varphi(\cG_{q^k}(1,s))=\cS , \,  \widehat{\cC}_2=\varphi(\cG_{q^k}(s-1,s))$$
and
the following result holds:

\begin{theorem}
The code $\widehat{\cC}$ defined in (\ref{eq:code(k,n-k)}) is an optimum distance flag code of type $(k,n-k)$ on $\bbF_q^n$ having the spread $\cS$ as a projected code.
\end{theorem}

\begin{proof}
Since $\widehat{\cC}_1=\varphi(\cG_{q^k}(1,s))=\cS$, we have $|\widehat{\cC}|=|\widehat{\cC}_1|=|\varphi(\cG_{q^k}(1,s))|=|\cG_{q^k}(1,s)|$. Furthermore, by the injectivity of $\varphi$ (see Proposition \ref{prop: field reduction}), we also have that $|\widehat{\cC}_2|=|\varphi(\cG_{q^k}(s-1,s))|=|\cG_{q^k}(s-1,s)|$. As $|\cG_{q^k}(1,s)|=|\cG_{q^k}(s-1,s)|$ we conclude that $|\widehat{\cC}|=|\widehat{\cC}_1|=|\widehat{\cC}_2|$ and $\widehat{\cC}$ is disjoint.  
 Let us now prove that the projected codes  of $\widehat{\cC}$ are constant dimension codes with the maximum possible distance. Since the projected code $\widehat{\cC}_1$ is a spread, it is enough to check this property for $\widehat{\cC}_2$. Given any two different subspaces $\varphi(H), \varphi(H') \in \widehat{\cC}_2=\varphi(\cG_{q^k}(s-1,s))$, by means of Proposition \ref{prop: field reduction}, we have that $H, H'$ are different hyperplanes. Moreover, since the intersection of any two hyperplanes in $\mathbb{F}_{q^k}^s$ is a $(s-2)$-dimensional subspace of $\bbF_{q^k}^s$ we have that
 $$
\dim(\varphi(H)\cap\varphi(H')) = \dim(\varphi(H\cap H')) =k(s-2)=n-2k.
$$
Notice that $n-2k= 2(n-k)-n$ is the minimum among the possible dimensions of the intersection of subspaces in $\cG_q(n-k,n)$. Hence $\widehat{\cC}_2$ is an equidistant $(n-2k)$-intersecting constant dimension code and, by applying Theorem \ref{theo:characterization optimum distance}, we are done.
\end{proof}
\subsubsection{The full admissible type vector}
In this subsection we finally tackle the construction of optimum distance flag codes of the full admissible type, that is, of type $(1, \dots, k, n-k, \dots, n-1)$ on $\bbF_q^{n}$ having the $k$-spread $\cS$ defined in (\ref{eq:ManganielloSpreadDefinition}) as a projected code. To do this, we start from the optimum distance flag code $\widehat{\cC}$ of type $(k,n-k)$ defined in (\ref{eq:code(k,n-k)}).  Recall that the construction of this code depends on the choice of a bijection $M : \cG_{q^k}(1,s) \rightarrow \cG_{q^k}(s-1,s)$ such that $l\subset M(l)$ for any $l\in \cG_{q^k}(1,s)$.

Let us fix an order in the set of lines of $\bbF_{q^k}^s$ and write $\cG_{q^k}(1,s)=\{l_1, l_2, \dots, l_L\}$, where $L=|\cG_{q^k}(1,s)|$. This order in $\cG_{q^k}(1,s)$ naturally induces respective orders in the sets $\cG_{q^k}(s-1,s)$, $\cS$ and $\cH=\varphi(\cG_{q^k}(s-1,s))$ as follows:
$$H_i=M(l_i), \, \cS_i=\varphi(l_i), \cH_i=\varphi(H_i)$$
for $i=1,\dots, L$.  Denote also by $\mathrm{S}_i $ the RREF generator matrix of $\cS_i$. Notice that $\mathrm{S}_i \in \Sigma$ where $\Sigma$ is the set defined in (\ref{eq:matrices for the spread}), and $\mathrm{S}_i=(\phi(x_{i1})| \dots| \phi(x_{is})),$ where $(x_{i1}, \dots, x_{is})\in \bbF_{q^k}^s$ is the  RREF matrix generating the line $l_i$. 

Now, given a hyperplane $H_i=M(l_i)$ of $\bbF_{q^k}^s$, we can write
$$
H_i = l_{i} \oplus  l_{i_2}  \oplus \dots \oplus l_{i_{s-1}},
$$ for $l_{i_2}, \dots, l_{i_{s-1}}$ some lines of $\bbF_{q^k}^s.$ By the properties of the field reduction $\varphi$ des\-cribed in Proposition \ref{prop: field reduction}, we have that 
$$
\cH_i=\varphi(H_i)=\cS_{i}\oplus \cS_{i_2}\oplus \dots \oplus \cS_{i_{s-1}}. 
$$
So, any subspace $\cH_i \in \cH$ can be decomposed as a direct sum of subspaces in $\cS$. This representation is not unique since $H_i$ can be written as direct sum of different collections of lines. Moreover, given that 
 for any line $l_{i_s}$ in $\bbF_{q^k}^s \setminus H_i$, it holds that $H_i\oplus l_{i_s}=\bbF_{q^k}^s$, by using Proposition \ref{prop: field reduction} again, we conclude that $\cH_i \oplus \cS_{i_s}=\bbF_q^{n}$. As a consequence, the rows of the matrix
\begin{equation}\label{def: matrices W}
\mathrm{W}_i=
\left(
\begin{array}{l}
\mathrm{S}_i\\
\mathrm{S}_{i_2}\\
\vdots\\
\mathrm{S}_{i_{s-1}}\\
\mathrm{S}_{i_s}\\
\end{array}
\right)
\end{equation} form a basis of $\bbF_q^{n}$. Moreover, any collection of $j\leq n$ rows of $\mathrm{W}_i$ generates a $j$-dimensional subspace of $\bbF_q^{n}$.

Denote by $\mathrm{W}_i^{(j)}$ the submatrix of $\mathrm{W}_i$  given by its first $j$ rows. We also denote by $\cW_i^{(j)}=\rsp(\mathrm{W}_i^{(j)})$. With this notation, it is clear that $\cW_i^{(k)}=\cS_i$ and $\cW_i^{(n-k)}=\cH_i$. In addition, for any $1\leq j_1 < j_2 \leq n$, it holds that $\cW_i^{(j_1)} \subsetneq \cW_i^{(j_2)}$. This fact allows us to define $\cF_{\mathrm{W}_i}$ the {\em flag of type $(1, \dots, k, n-k, \dots, n-1)$ associated to $\mathrm{W}_i$ } in the following way:

\begin{equation}\label{eq: flag associated to W_i}
\cF_{\mathrm{W}_i}=(\cW_i^{(1)}, \dots, \cW_i^{(k-1)}, \cS_i, \cH_i, \cW_i^{(n-k+1)}, \dots, \cW_i^{(n-1)}).
\end{equation}
Finally, given the family of matrices $\{W_i\}_{i=1}^L$, we define the family of associated flags of type $(1,\dots,k,n-k,\dots,n-1)$:

\begin{equation}\label{def: optimum distance flag code}
\cC=\{ \cF_{\mathrm{W}_i}  \ | \ i=1, \dots, L \}.
\end{equation}

Let us see that $\cC$ is an optimum distance flag code. To do so, we analyze the structure of  its projected codes:

\begin{equation}\label{def: pojected before k}
\cC_j=\{ \cW_i^{(j)} \ | \ i=1, \dots, L \}
\end{equation}
and
\begin{equation}\label{def: pojected after k}
\cC_{k+j}=\{ \cW_i^{(n-k+j-1)} \ | \ i=1, \dots, L \},
\end{equation}
for all $j=1, \dots, k.$

\begin{proposition}\label{prop: optimum distance projected code}
Given the flag code $\cC$ defined as above, for each $j=1, \dots, k$ the following is satisfied:
\begin{enumerate}
\item The code $\cC_j$ is a partial spread in $\cG_q(j, n)$  with cardinality $L=\frac{q^n-1}{q^k-1}$.
\item The code $\cC_{k+j}$ is an equidistant $(n-2k+2(j-1))$-intersecting constant dimension code in $\cG_q(n-k+j-1, n)$ with cardinality $L=\frac{q^n-1}{q^k-1}$. As a consequence, $\cC_{k+j}$ is a constant dimension code of maximum distance. \label{theo: projected after k}
\end{enumerate}
In particular, we have that $\cC_k=\cS$ and $\cC_{k+1}=\cH$.
\end{proposition}
\begin{proof}
By construction it is clear that $\cC_k=\cS$ and $\cC_{k+1}=\cH$. Now, for any $1\leq j\leq k$, given two different indices $i_1,i_2\in \{ 1, \dots, L\}$, we have that
$$
\cW_{i_1}^{(j)} \cap \cW_{i_2}^{(j)} \subset \cS_{i_1} \cap \cS_{i_2} = \{0\}.
$$Hence, $\cC_j$ is a partial spread in the Grassmannian $\cG_q(j,n)$ with $|\cC_j|=L$.

To prove (\ref{theo: projected after k}), consider subspaces $\cW_{i_1}^{(n-k+j-1)}, \cW_{i_2}^{(n-k+j-1)} \in \cC_{k+j}$.  We know that $\dim(\cH_{i_1} \cap \cH_{i_2})=(s-2)k=n-2k$, then the subspace sum $\cH_{i_1} +\cH_{i_2}$ is the whole space $\bbF_q^{n}$. As a consequence, 
$$
n = \dim(\cH_{i_1} +\cH_{i_2}) \leq \dim(\cW_{i_1}^{(n-k+j-1)}+\cW_{i_2}^{(n-k+j-1)})\leq n
$$
and then it follows that
$$
\begin{array}{ccl}
n & = & \dim(\cW_{i_1}^{(n-k+j-1)} + \cW_{i_2}^{(n-k+j-1)})\\
  & = & 2(n-k+j-1)-\dim(\cW_{i_1}^{(n-k+j-1)} \cap\cW_{i_2}^{(n-k+j-1)}).
\end{array}
$$
Hence, we obtain
$$
\begin{array}{ccl}
\dim(\cW_{i_1}^{(n-k+j-1)} \cap\cW_{i_2}^{(n-k+j-1)}) & = & 2(n-k+j-1)-n \\
                                                      & = & n-2k+2(j-1),
\end{array}
$$ which is the minimum possible dimension of the intersection between subspaces of dimension $n-k+j-1$ of $\bbF_q^{n}$. Thus, we conclude that $\cC_{k+j}$ is an equidistant $(n-2k+2(j-1))$-intersecting constant dimension code with exactly $L$ elements. In particular, we have that  $d_S(\cC_{k+j})= 2(k-(j-1))$ and $\cC_{k+j}$ is a constant dimension code with the maximum distance.
\end{proof}

\begin{theorem}
The flag code $\cC$ defined in (\ref{def: optimum distance flag code}) is an optimum distance flag code of type $(1, \dots, k, n-k, \dots, n-1)$ on $\bbF_q^{n}$ with the $k$-spread $\cS$ as a $k$-projected code. This code has cardinality $|\cC|=L=\frac{q^n-1}{q^k-1}$ and distance $d_f(\cC)=2k(k+1)$.
\end{theorem}
\begin{proof}
By means of Proposition \ref{prop: optimum distance projected code} we conclude that $\cC$ is a disjoint flag code of cardinality $L$ with projected codes attaining the maximum distance for their corresponding dimensions. Then, by Theorem \ref{theo:characterization optimum distance}, $\cC$ is an optimum distance flag code, that is,  $d_f(\cC)=2k(k+1)$.
\end{proof}
\begin{remark}
The code $\cC$ defined in (\ref{def: optimum distance flag code}) attains the maximum possible distance for flag codes of type $(1, \dots, k, n-k, \dots, n-1)$ on $\bbF_q^{n}$. Furthermore, by means of Theorem \ref{theo: maximum cardinality iff spread},  it also has the best possible size among the optimum distance flag codes of the full admissible type vector on $\bbF_q^n$.

\end{remark}

\subsubsection{The general case}

Finally, in order to get an optimum distance flag code of any admissible type vector with a $k$-spread as a projected code, we apply a \emph{puncturing process} to the code $\cC$ defined in (\ref{def: optimum distance flag code}). This process was already used in \cite{CasoPlanar} to get optimum distance flag codes having a planar spread as a projected code. Let us recall it. Fix an admissible type vector $(t_1, \dots, t_r)$, that is, a type vector such that $k \in \{t_1,\dots, t_r\}\subseteq\{1,\dots, k, n-k,\dots, n-1\}$. Consider a flag $\cF_{\mathrm{W_i}}$ in the code $\cC$ in (\ref{def: optimum distance flag code}). The punctured flag of type $(t_1, \dots, t_r)$ associated to $\cF_{\mathrm{W_i}}$ is the sequence
\begin{equation}\label{def: punctured flag}
(\cW_i^{(t_1)}, \dots, \cW_i^{(t_r)}).
\end{equation}
The  punctured flag code of type $(t_1, \dots, t_r)$ associated to $\cC$ is the code given by
\begin{equation}\label{def: punctured flag code}
\cC_{(t_1, \dots, t_r)}= \{ (\cW_i^{(t_1)}, \dots, \cW_i^{(t_r)}) \ | \ i=1, \dots, L\}.
\end{equation}
Observe that the projected codes of $\cC_{(t_1, \dots, t_r)}$ are, in particular, projected codes of $\cC$. Hence, the next result follows straightforwardly from this fact, together with Theorem \ref{theo: maximum cardinality iff spread}.

\begin{theorem}
Given $n$ and a divisor $k$ of $n$, for every admissible type vector $(t_1, \dots, t_r)$, the code $\cC_{(t_1, \dots, t_r)}$ defined as above is an optimum distance flag code on $\bbF_q^n$ with the spread $\cS$  as a projected code. Its cardinality, which is $L=\frac{q^n-1}{q^k-1}$, is maximum for optimum distance flag codes of this type.
\end{theorem}

\subsection{Example}\label{subsec:example}
We construct an example of optimum distance flag code of type $(2,4)$ on $\bbF_2^6$ from a $2$-spread. To do this, we follow the steps given in Section \ref{sec:our construction}.

Consider the bipartite  graph $G=(V, E)$ where $V=\cG_4(1,3) \cup \cG_4(2,3)$ and $E$ is the set of pairs $(l,H) \in \cG_4(1,3)\times \cG_4(2,3)$ with $l\subset H$. 
Take $\alpha \in \bbF_4$ with $\alpha\neq 0,1$. Then, we have that $\bbF_4=\{0,1,\alpha, \alpha^2\}$.  By using the package GRAPE of GAP, and following the process described in \cite{Diestel2005} to get perfect matchings, we have designed an algorithm that provides a perfect matching of $V$. The induced bijection $M : \cG_4(1,3)  \rightarrow \cG_4(2,3)$ is explicitly given by:
\begin{multicols}{2}\small{$$
\begin{array}{ccc}
M(\left\langle (0,0,1) \right\rangle) & = & \mathrm{rowsp}
\begin{pmatrix}
0 & 0  & 1  \\
0 & 1  & 0
\end{pmatrix} \qquad \\
M(\left\langle (0,1,0) \right\rangle) & = & \mathrm{rowsp}
\begin{pmatrix}
 0 & 1  & 0  \\
 1 & 0  & \alpha^2
\end{pmatrix} \\
M(\left\langle (0,1,1 ) \right\rangle) & = & \mathrm{rowsp}
\begin{pmatrix}
 0 & 1  &  1 \\
 1 & 0  & 1
\end{pmatrix}\\
M(\left\langle (0,1,\alpha ) \right\rangle) & = & \mathrm{rowsp}
\begin{pmatrix}
  0 &  1 & \alpha  \\
  1 &  0 & \alpha
\end{pmatrix} \\
M(\left\langle ( 0,1,\alpha^2 ) \right\rangle) & = & \mathrm{rowsp}
\begin{pmatrix}
 0 & 1  &  \alpha^2 \\
 1 & 0  &  1
\end{pmatrix} \\
M(\left\langle ( 1,0,0 ) \right\rangle) & = & \mathrm{rowsp}
\begin{pmatrix}
 1 & 0 & 0   \\
 0 & 1  & \alpha^2  
\end{pmatrix}\\
M(\left\langle ( 1,0,1 ) \right\rangle) & = & \mathrm{rowsp}
\begin{pmatrix}
1 & 0 & 1  \\
0 & 0 & 1 
\end{pmatrix}\\
M(\left\langle ( 1,0,\alpha ) \right\rangle) & = & \mathrm{rowsp}
\begin{pmatrix}
 1 &  0 &  \alpha  \\
 0 &  1 &  1 
\end{pmatrix}\\
M(\left\langle ( 1,0,\alpha^2 ) \right\rangle) & = & \mathrm{rowsp}
\begin{pmatrix}
1 & 0 & \alpha^2   \\
0 & 1 &  \alpha 
\end{pmatrix}\\
M(\left\langle ( 1,1,0 ) \right\rangle) & = & \mathrm{rowsp}
\begin{pmatrix}
 1 &  1 &  0  \\
 0 &  1 &  \alpha^2 
\end{pmatrix}\\
M(\left\langle ( 1,1,1 ) \right\rangle) & = & \mathrm{rowsp}
\begin{pmatrix}
 1 & 1 & 1   \\
 0 & 1 & \alpha^2
\end{pmatrix}
\end{array}
$$

\columnbreak

$$
\begin{array}{ccc}
M(\left\langle ( 1,1,\alpha ) \right\rangle) & = & \mathrm{rowsp}
\begin{pmatrix}
 1 & 1  &  \alpha  \\
 0 & 1  &  \alpha
\end{pmatrix}\\
M(\left\langle ( 1,1,\alpha^2 ) \right\rangle) & = & \mathrm{rowsp}
\begin{pmatrix}
 1 &  1 & \alpha^2   \\
 0 &  0 &  1
\end{pmatrix}\\
M(\left\langle ( 1,\alpha,0 ) \right\rangle) & = & \mathrm{rowsp}
\begin{pmatrix}
 1 & \alpha &  0  \\
 0 & 1      &  0  
\end{pmatrix}\\
M(\left\langle ( 1, \alpha, 1 ) \right\rangle) & = & \mathrm{rowsp}
\begin{pmatrix}
1 & \alpha & 1   \\
0 & 1      & 1 
\end{pmatrix}\\
M(\left\langle ( 1, \alpha, \alpha ) \right\rangle) & = & \mathrm{rowsp}
\begin{pmatrix}
 1 & \alpha & \alpha   \\
 0 & 1      & 0 
\end{pmatrix}\\
M(\left\langle ( 1, \alpha, \alpha^2 ) \right\rangle) & = & \mathrm{rowsp}
\begin{pmatrix}
 1 & \alpha & \alpha^2    \\
 0 &  0     &  1
\end{pmatrix}\\
M(\left\langle ( 1,\alpha^2, 0 ) \right\rangle) & = & \mathrm{rowsp}
\begin{pmatrix}
 1 & \alpha^2 & 0   \\
 0 &  1       & \alpha 
\end{pmatrix}\\
M(\left\langle ( 1, \alpha^2, 1 ) \right\rangle) & = & \mathrm{rowsp}
\begin{pmatrix}
 1 & \alpha^2 & 1   \\
 0 & 1        & 0 
\end{pmatrix}\\
M(\left\langle (1, \alpha^2, \alpha  ) \right\rangle) & = & \mathrm{rowsp}
\begin{pmatrix}
1 & \alpha^2 & \alpha   \\
0 & 0 & 1 
\end{pmatrix}\\
M(\left\langle ( 1, \alpha^2, \alpha^2 ) \right\rangle) & = & \mathrm{rowsp}
\begin{pmatrix}
 1 & \alpha^2 & \alpha^2   \\
 0 &  1       & 1 
\end{pmatrix}\\
& & 
\end{array}
$$}
\end{multicols}

Observe that every line $l\in \cG_4(1,3)$ is a subspace of the (hyper)plane $M(l)$. Even more, we have expressed every subspace $M(l)$ as the rowspace of a $2\times 3$ matrix whose first row is precisely a generator of the line $l$. In this way, we obtain the optimum distance flag code of type $(1,2)$ on $\mathbb{F}_4^3$
$$\widetilde{\cC}=\{(l, M(l))\mid l \in \cG_4(1,3)\}.$$

Now, let $f(x) \in\bbF_2[x]$ be the minimal polynomial of $\alpha$, which has degree $2$, and $P \in GL_2(2)$ its companion matrix. If $\phi$ is the field isomorphism  in (\ref{eq:field isomorphism}), we have that $\phi(0)={\0}_2$, $\phi(1)=I_2$ and $\phi(\alpha)=P$. Taking the previous matching $M$ and the field reduction $\varphi$ induced by $\phi$ (\ref{eq: field reduction on m-subspaces}), we define the following optimum distance flag code of type $(2,4)$ on $\mathbb{F}_2^6$

$$\widehat{\cC}=\{(\varphi(l), \varphi(M(l))) \mid  l \in \cG_4(1,3)\}.$$
If we take $l=\langle(0,1,\alpha)\rangle$, for instance, the corresponding element on $\widehat{\cC}$ is the flag
$$\cF=  \begin{array}{c}
 \left(\textrm{rowsp} \begin{pmatrix} {\0}_2 & I_2& P  \end{pmatrix}, \, \textrm{rowsp}
\begin{pmatrix}
 {\0}_2 & I_2& P \\
  I_2 &  {\0}_2 & P
\end{pmatrix}\right).
\end{array}$$
Note that $\widehat{\cC}_1=\cS(3,2,P)=\cS$. Also, for every  $\varphi(l) \in \cS$ with $ l \in \cG_4(1,3)$, we have that  $\varphi(M(l))$ is a $4$-dimensional subspace over $\bbF_2$ that contains $\varphi(l)$. Moreover,   $\varphi(M(l))$ is the vector space generated by the rows of a $4\times 6$ full-rank matrix, whose two first rows span $\varphi(l)$.
\section{Conclusions and future work}\label{sec: conclusions}

In this paper we have addressed the problem of obtaining  flag codes of general type $(t_1,\dots,t_r)$ on a space $\mathbb{F}_q^n$ with the maximum possible distance and the property of having a $k$-spread as a projected code whenever $k$ divides $n$. Firstly, we have showed that the existence of such codes might be not possible for an arbitrary type vector and  have characterized the admissible ones. They have to satisfy the condition: $k \in \{t_1,\dots, t_r\} \subseteq\{1,\dots, k,n-k,\dots, n-1\}$.

Given an admissible type vector, we have proved the existence of optimum distance flag codes of such a type with a spread as a projected code by describing a gradual construction starting from type $(1, n-1)$, following with type $(k, n-k)$, to finish with the full admissible type $\{1,\dots, k,n-k,\dots, n-1\}$. This construction is mainly based on two ideas: on one side, we exploit the existence of perfect matchings in the bipartite graph with set of vertices given by the lines and the hyperplanes of $\mathbb{F}_q^n$ and edges given by the containment relation. On the other hand, we use the properties of the field reduction map that allow us to translate the spread of lines to a $k$-spread and to build our code from it. Our construction provides codes with the best possible size among optimum distance flag codes of any arbitrary admissible type vector.

In current work we investigate the algebraic structure and features of this family of codes and explore other possible constructions. We also study the family of flag codes from spreads not necessarily having the maximum distance as well as the existence and performance of decoding algorithms for them.

\end{document}